\documentclass[smallcondensed]{svjour3}

\usepackage{hyperref}
\usepackage{amssymb}
\usepackage{amsmath}

\usepackage{mathabx}
\usepackage{tikz}
\usepackage{url}

\makeatletter
\newcommand{\labitem}[2]{%
\def\@itemlabel{\textbf{#1}}
\item
\def\@currentlabel{#1}\label{#2}}
\makeatother

\title{A generalized Sitnikov problem                                                                                                                                                                                                                                                                                                                                                                            }

\author{Gast\'on Beltritti   \and Fernando Mazzone   \and Martina Oviedo  }

\institute{G. Beltritti \at CONICET - Dpto. de Matem\'atica, Facultad de Ciencias Exactas Físico-Químicas y Naturales.
Universidad Nacional de R\'{i}o Cuarto
(5800) R\'{\i}o Cuarto, C\'ordoba, Argentina\\
\email{gbeltritti@exa.unrc.edu.ar}\\
F. Mazzone \at CONICET - Dpto. de Matem\'atica, Facultad de Ciencias Exactas Físico-Químicas y Naturales.
Universidad Nacional de R\'{i}o Cuarto
(5800) R\'{\i}o Cuarto, C\'ordoba, Argentina\\
\email{fmazzone@exa.unrc.edu.ar}\\
M. Oviedo \at
 CONICET - Instituto de Investigaciones Matem\'aticas ``Luis A. Santal\'o''.
 Facultad de Ciencias Exactas y Naturales-UBA.
 (C1428EGA) – C.A.B.A., Argentina.\\
\email{ moviedo@itba.edu.ar}
}

\newcommand{\rr}{\mathbb{R}}

\begin{document}

\maketitle

\begin{abstract}
In this paper we address a $n+1$-body gravitational problem governed by the Newton's laws, where $n$ primary bodies orbit on a plane $\Pi$ and an additional massless particle moves on the perpendicular line to $\Pi$ passing through the center of mass of the primary bodies. We find a condition for that the configuration described be possible. In the case that the primaries are in a rigid motion we classify all the motions of the massless particle. We study the situation when the massless particle has a periodic motion with the same minimal period than primary bodies. We show that this fact is related with the existence of certain pyramidal central configuration.
\end{abstract}

\section{Introduction}
In this paper we study the following restricted  Newtonian $n+1$-body problem $P$ (see figure \ref{fig:conf_esp}):
\begin{itemize}
 \item[$P_1$] We have $n$ primary bodies of masses $m_1,\ldots,m_n$ and an additional massless particle.
 \item[$P_2$] The primary bodies are in a homographic motion (see \cite[Section 2.9]{JaumeLlibre276}). This motion is carried out in a plane $\Pi$.
 \item[$P_3$] The massless particle is moving  on the perpendicular line to $\Pi$ passing through the center of mass of the primary bodies.
\end{itemize}

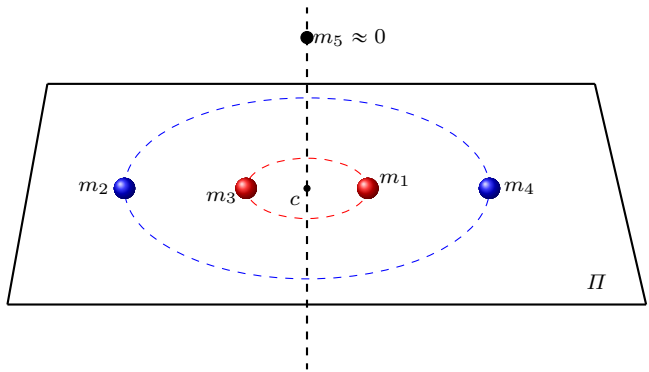
\begin{figure}[h]
\begin{center}
\begin{tikzpicture}[xscale=.8, yscale=.8]
				\draw[black,  line width=.8pt] (-6,0,-4.5)--(3,0,-4.5);
				\draw[black, line width=.8pt] (-3,0,5)--(7.5,0,5);
				\draw[black, line width=.8pt] (-3,0,5)--(-6,0,-4.5);
				\draw[black, line width=.8pt] (3,0,-4.5)--(7.5,0,5);
				\draw[red,dashed] (0,0,0) ellipse (1cm and 0.5cm);
				\draw[blue,dashed] (0,0,0) ellipse (3cm and 1.5cm);
					\shade[ball color=red]  (1,0) circle (5 pt);
					\node at (1.45,0.14) {$m_1$};
					\shade[ball color=red]  (-1,0) circle (5 pt);
					\node at (-1.4,-0.15) {$m_3$};

					\shade[ball color=blue]  (3,0) circle (5 pt);
					\node at (-3.5,0) {$m_2$};
					\shade[ball color=blue]  (-3,0)
					 circle (5 pt);
					\node at (3.5,0) {$m_4$};

        \draw[black, dashed, line width=.8pt] (0,0,0)--(0,3,0);
		\draw[black, dashed, line width=.8pt] (0,0,0)--(0,-3,0);
		\draw[fill=black](0,2.5,0) circle (0.1 cm);
		\node at (0.7,2.5,0) {$m_5\approx 0$};
				\node at (6.3,0,4) {$\Pi$};

				\draw[fill=black](0,0,0) circle (0.05 cm);
				\node at (-0.2,-0.2,0) {$c$};
\end{tikzpicture}\caption{Five-body problem with primaries in a collinear configuration}\label{fig:conf_esp}
\end{center}
\end{figure}

Problems like the one presented above have been extensively discussed in the literature. In \cite{sitnikov1960existence} K. Sitnikov considered the problem of two body in a Keplerian elliptic motion and a massless particle moving in the perpendicular line to the orbital plane passing
through the center of mass. Sitnikov obtained deep results about existence of solutions, for small $e>0$, with a chaotic behavior (see \cite[III(5)]{moser2016stable}). Periodic solutions for a Sitnikov configuration were considered in
\cite{corbera2000periodic,corbera2002symmetric,llibre2008families,pustyl1990certain}.

Generalized circular Sitnikov problems, i.e. we have $n\geq 3$ primaries in a relative equilibrium motion,   were addressed more recently.
In \cite{soulis2008periodic} Soulis, Papadakis and Bountis studied existence, linear stability and bifurcations for a problem similar to $P$. They considered  a Lagrangian equilateral triangle configuration for the primary bodies, which were supposed to have the same mass $m_1=m_2=m_3$. In \cite{bountis2009stability} Papadakis and Bountis extended the results of \cite{soulis2008periodic} to $n$ primaries ($n\geq 3$) in a poligonal equal masses configuration. Later  in \cite{pandey2013periodic}, Pandey and Ahmad generalized  the analysis started in \cite{soulis2008periodic} to the case with oblate primaries.
In \cite{li2013characterization} Li, Zhang and Zhao studied a special type of
restricted circular $n+1$-body problem  with equal masses for the primaries in a regular polygon configuration. Periodic solutions for generalized Sitnikov problems with primaries performing  no rigid motions were studied in \cite{pustyl1990certain,rivera2013periodic}. We emphasize that in
\cite{bountis2009stability,li2013characterization,pandey2013periodic,pustyl1990certain,rivera2013periodic,soulis2008periodic} it is supposed that
the primary bodies are in the vertices of a regular polygon.
As far as we know the first non-polygonal configurations of primary bodies was considered in \cite{marchesin2013spatial} where  Marchesin and Vidal studied the problem $P$ for a rigid motion  of primaries in a  rhomboidal configuration.
 In \cite{bakker2015separating} Bakker and Simmons studied scape regions for the massless particle in a problem similar to $P$ where the primaries performing certain type of periodic orbits including non homographic motions.

In the present paper, after introducing preliminaries facts in Section \ref{sec:pre},   we
obtain in Section \ref{sec:admisible.configuraciones} necessary and sufficient conditions on the configuration of primary bodies in order  that the $z$-axis to be invariant for the flow associated to the motion equations of the massless particle. For this type of configurations, that we call \emph{balanced}, the Sitnikov problem has sense.  In Section \ref{sec:addmisibles} we will find all balanced configurations for $n\leq 4$ primaries.  The Section \ref{sec:mas-mot} is devoted to describe all possible motion of the massless particle when the primaries are in a relative equilibrium (or rigid) motion. In this direction we observe that only are possible scape (both parabolic and hiperbolic) and periodic motions. Moreover, we will give a formula expressing the period of solutions  by means of integrals.  We prove in Corollary \ref{cor:sol.periodica.sist.completo} that the complete $n+1$-body system has  infinite quantity of periodic solutions. In Section  \ref{sec:sincro} we discuss the
situation when the entire system has a solution with the same period that the rigid motion of primaries. We call it \emph{synchronous solution}. Surprisingly the existence of synchronous solutions is related to the existence of certain pyramidal central configurations (for the definition of this concept see \cite{fayccal1996classification,faycaltesis,ouyang2004pyramidal}). Finally, in the last section, we study certain non balanced configurations which allows some particular solutions of problem $P$.

In this paper we generalize and extend some results previously obtained. For example, our results in Section \ref{sec:mas-mot} concerning to balanced configurations   generalize the results  in \cite{marchesin2013spatial} established for rhomboidal configurations. In Section \ref{sec:sincro} we prove that there exists synchronous solutions for primaries in a regular poligonal equal mass configuration if and only if $2\leq n\leq 472$. The sufficient of this fact was established in \cite{li2013characterization}.

\section{Preliminaries}\label{sec:pre}

We start considering $n$  mass points, $n>2$, of masses $m_1,\ldots,m_n$ moving in a Euclidean 3-dimensional space according to Newton's laws of motion. We assume that $x_1(t),\ldots,x_n(t)$ are the coordinates of the bodies in some inertial Cartesian coordinate system.  We can suppose, without any loss generality, that the center of mass   $C:=\sum_jm_jx_j/M$ ($M:=\sum_j m_j$) is fixed at the origin ($C=0$).

Initially we assume that the bodies are in a \emph{planar homographic motion} on the plane $\Pi$ (see \cite{JaumeLlibre276}), where it is assumed that $\Pi$ is the plane determined by the first two coordinates axes. Concretely  we are assuming  that

\begin{equation}\label{eq:x_j=rtQtq_j}
 x_j(t)=r(t)Q(\theta (t))q_j,
\end{equation}
where
\[
 Q(\theta )=\begin{pmatrix}
           \cos(\theta ) & -\sin(\theta ) & 0\\
           \sin(\theta ) & \cos(\theta ) & 0\\
           0            &     0     &  1\\
          \end{pmatrix}
\]
and $q_j\in\Pi$, $j=1,\ldots,n$ are vectors in a planar \emph{central configuration} (CC) in $\Pi$. We recall the following definition of this concept (see \cite{JaumeLlibre276}).

\begin{definition}\label{def:CC}
Let $q=(q_1,\ldots,q_n)$ be  a n-tuple of positions in $\rr^3$ and let $m=(m_1,\ldots,m_n)$ be a vector of masses. We say that $(q,m)$ is a central configuration if
there exists $\lambda\in\rr$ such that
\begin{equation}\label{eq:def.CC}
\nabla_jU(q_1,\ldots,q_n)+\lambda m_jq_j=0,\quad j=1,\ldots,n.
\end{equation}
where
\begin{equation}\label{eq:potencial}
 U(q_1,\ldots,q_n)=\sum_{i<j}\frac{m_im_j}{r_{ij}},
\end{equation}
 $r_{ij}=|q_i-q_j|$ and $\nabla_j$ denotes the $3$-dimensional partial gradient with respect to $q_j$.
\end{definition}

According to \cite[Eq. (2.16)]{JaumeLlibre276}  the functions $r(t)$ and $\theta (t)$ solves the two-dimen\-sional Kepler problem in polar coordinates, i.e.
\begin{equation}
 \begin{array}{rl}\label{eq:kepler.2.dim}
\ddot{r}(t)-r(t)\dot{\theta}(t)^2 & = -\frac{\lambda}{r(t)^2}\\
\frac{d }{dt}\left[ r(t)^2\dot{\theta}(t)\right] & =0.\\
\end{array}
\end{equation}
In the particular case of \emph{rigid motion}, we have  $r(t)\equiv 1$ and $\theta (t)=\sqrt{\lambda }t +\theta(0)$. In this case the primary bodies perform a periodic motion with minimal period $T:=2\pi/\sqrt{\lambda }$.

Let $x_0(t)$ be the position of the massless particle.
According to the Newtonian equations of motion $x_0$ satisfies
\begin{equation}\label{eq:newton}
 \ddot{x}_0=\sum_{i=1}^n\frac{m_i(x_i-x_0)}{|x_i-x_0|^3}=:f(t,x_0).
\end{equation}

In the previous equation, we assume know the positions of the primaries. Therefore, this equation plus  initial conditions completely determines the position of the particle.

\section{Balanced configurations}\label{sec:admisible.configuraciones}
Henceforth we denote by $L$ the coordinate $z$ axis.
It is well know that a  necessary and sufficient condition for that $L$ be invariant under the  flow associated to the non autonomous system  \eqref{eq:newton}, is that $f(t,L)\subset L$ for all $t$, i.e. $L$ is \emph{$f$-invariant} for every $t$.

\begin{definition}
We say that a
central configuration $(q,m)$ is \emph{balanced} if and only if $(q,m)$ satisfies that, for any $r>0$, such that the set
\[F_r:=\{i:|q_i|=r\}\]
is non empty, then
\begin{equation}\label{eq:suma0}\sum_{i\in F_r}m_iq_i=0.\end{equation}
i.e. every maximal set of  bodies which are equidistant from origin has center of mass equal to $0$.
\end{definition}

\begin{theorem}\label{thm:prim} $L$ is $f$-invariant for every $t$ if and only if $(q,m)$ is balanced.
\end{theorem}

For the proof of the previous theorem we need the following result.

\begin{lemma}\label{lem:1} For $c>0$ we define the function $y_c(t):=(c+t)^{-3/2}$. If $0<t_1<t_2<\ldots<t_k$ then the functions $y_j(t):=y_{t_j}(t)$  are linearly independent on  each open interval   $I\subset \mathbb{R}^+$.
\end{lemma}
\begin{proof} It is sufficient to prove that the Wronskian

 \[W:=W(y_1,\ldots,y_k)(t)=\det\begin{pmatrix}
			      y_1 & \cdots & y_k\\
			      \frac{dy_1}{dt}&  \cdots & \frac{dy_k}{dt}\\
			      \vdots & \ddots & \vdots \\
			      \frac{d^{k-1}y_1}{dt^{k-1}}&  \cdots & \frac{d^{k-1}y_k}{dt^{k-1}}\\
                           \end{pmatrix}
\]
is not null on $I$.

Using induction is easy to show that
\begin{equation}\label{eq:der_ind}\frac{d^iy_c}{dt^i}=\beta_{i}y_{c}^{\frac{2i+3}{3}},\quad\hbox{for some }\beta_{i}\neq 0, \hbox{ and for all }i=1,\ldots.
\end{equation}
Fix any $t\in I$. Then, according to \eqref{eq:der_ind} and writing $\lambda_j:=(t+t_j)^{-1}$, we have

\[
\begin{split}
  W(t)&=\det
    \begin{pmatrix}
      \lambda_1^{3/2} & \lambda_2^{3/2} &\cdots & \lambda_k^{3/2} \\
      \beta_1\lambda_1^{5/2} &\beta_1 \lambda_2^{5/2} &\cdots &\beta_1 \lambda_k^{5/2}\\
      \vdots & \vdots &\ddots & \vdots\\
      \beta_{k-1}\lambda_1^{k+1/2} & \beta_{k-1}\lambda_2^{k+1/2} &\cdots & \beta_{k-1}\lambda_k^{k+1/2}
    \end{pmatrix}
  \\
  &= \beta_1\beta_2\cdots\beta_{k-1} \lambda_1^{3/2}\lambda_2^{3/2}\cdots \lambda_k^{3/2}
     \det \begin{pmatrix}
      1& 1 &\cdots & 1 \\
      \lambda_1 & \lambda_2 &\cdots & \lambda_k\\
      \vdots & \vdots &\ddots & \vdots\\
      \lambda_1^{k-1} & \lambda_2^{k-1} &\cdots & \lambda_k^{k-1}
    \end{pmatrix}
  \\
  &= \beta_1\beta_2\cdots\beta_{k-1} \lambda_1^{3/2}\lambda_2^{3/2}\cdots \lambda_k^{3/2}
  \prod_{1\leq i<j\leq n}(\lambda_j-\lambda_i)
,
\end{split}
\]
where the last equality follows of the well known Vandermonde determinant identity. Therefore $W\neq 0$ if and only if $\lambda_i\neq\lambda_j$, $i\neq j$,
which in turn is equivalent to $t_i\neq t_j$, $i\neq j$.
\end{proof}

\begin{proof}[Proof of Theorem \ref{thm:prim}]
The condition $f(t,L)\subset L$ for all $t$ is equivalent to
\begin{equation}\label{eq:f(L)cL-->sum=0}
 \sum_{i=1}^n\frac{m_ir(t)Q(\theta (t))q_i}{\left(r(t)^2|q_i|^2+z^2\right)^{3/2}}=0\in\rr^2,
\end{equation}
for every $t,z\in \rr$.

Let $D=\{|q_i|: i=1,\ldots,n\}$.  Suppose that $D=\{s_1,\ldots,s_k\}$, with $s_i\neq s_j$ for $i\neq j$.  Therefore $\{1,\ldots,n\}=F_{s_1}\cup \cdots\cup F_{s_k}$. Then, multiplying equation \eqref{eq:f(L)cL-->sum=0} by $r(t)^2Q^{-1}(\theta(t))$  and writing $\zeta=(z/r(t))^2$ we have that \eqref{eq:f(L)cL-->sum=0} is equivalent to

\[\sum_{j=1}^k\left\{\frac{1}{(s_j^{2}+\zeta)^{3/2}}\sum_{i\in F_{s_j}}m_iq_i\right\}=0.\]
According to Lemma \ref{lem:1}, the last equation is equivalent to \eqref{eq:suma0}.
\end{proof}

\section{Balanced collisionless configurations for $n\leq 4$}\label{sec:addmisibles}

In this section we find all balanced collisionless configurations with $n\leq 4$. Hereafter  we said that   $q_1,\ldots,q_n$ is a collisionless configuration  when $q_i\neq 0$ for $i=1,\ldots,n$. We note that despite this fact the system can have collisions, for example in the case when the primaries have a homothetic collapsing motion.  Since the center of mass is an excluded position a balanced configuration satisfies
\begin{equation}\label{F_r.no.vac-->CF_r.geq2}
 F_r\neq \emptyset \Rightarrow \# F_r\geq 2.
\end{equation}

It is atrivial fact that  two point masses $m_1$ and $m_2$ is balanced if and only if $m_1=m_2$.

From \eqref{F_r.no.vac-->CF_r.geq2}, a $3$-body balanced configuration consists of equidistant  bodies from the origin. Therefore, it must to be the Lagrangian equilateral triangle. Now, by equation \eqref{eq:suma0} and an elementary geometrical reazoning   we have that $m_1=m_2=m_3$.

The case $n=4$ is more interesting. We include Definition \ref{def:bis.per} and Theorem \ref{thm:bisector.moeckel}, which are presented for the first time in  \cite{moeckel1990central}, for the reader's convenience.

\begin{definition}\label{def:bis.per}
Let $q$ be a planar configuration. For each pair, $i$, $j$, the line
containing $q_i$ and $q_j$ together with its perpendicular bisector form axes which
divide the plane into four quadrants. The union of the first and third quadrants
is an hourglass shaped region which will be called a `cone'; similarly,
the second and fourth quadrants together form another cone. The phrase `open
cone' refers to a cone minus the axes.
\end{definition}

\begin{theorem}[Perpendicular Bisector Theorem]\label{thm:bisector.moeckel}
Let $(q,m)$ be a planar central configuration and let
$q_i$ and $q_j$ be any two of its points. Then if one of the two open cones determined
by the line through $q_i$ and $q_j$ and its perpendicular bisector contains points of
the configuration, so does the other one.
\end{theorem}

Next we characterize all the $4$-body balanced collisionless configurations.

\begin{theorem}\label{thm:caracterizacion4}
Let $(q,m)$ be a 4-body central configuration. Then $(q,m)$ is  balanced and collisionless if and only if,  for a suitable enumeration  of bodies,   $q_1=-q_3$, $q_2=-q_4$, $m_1=m_3$,  $m_2=m_4$, and  $(q,m)$ is of some of the following mutually exclusive types:
\begin{description}
\item[CCcl.]   collinear,
\item[CCr.]  a rhombus with $r_{13}<r_{24}$ and $m_1>m_2$,
\item[CCs.]  a square with four equal masses.
\end{description}
\end{theorem}

\begin{remark} In \cite{shoaib2011collinear} was studied central configurations of type CCcl, while that CCr configurations were addressed in \cite{long2002four} and \cite{perez2007convex}.

\end{remark}

\begin{proof}
From \eqref{F_r.no.vac-->CF_r.geq2} we have to consider two cases.

\emph{Case 1.}  $m_1\geq m_2$, $|q_1|\neq|q_2|$, $|q_1|=|q_3|$ and $|q_2|=|q_4|$. Now \eqref{eq:suma0} implies that
 $m_1=m_3$, $m_2=m_4$, $q_1=-q_3$ and $q_2=-q_4$.  We divide the plane in two open cones $C_i$, $i=1,2$, by means of  the line $P$ joining $q_1$  and $q_3$ together with its perpendicular bisector $M$.  From Theorem \ref{thm:bisector.moeckel}, if  $q_2$  is in $C_1$, then  $q_4$ is in $C_2$, and vice versa. This is a contradiction with the fact that $q_2=-q_4$. Then $q_2,q_4\in P$ or $q_2,q_4\in M$, i.e. $q$ is collinear or is a rhombus with equal masses in opposite vertices. In the first case, $(q,m)$ is of  CCcl type. In the second case, if $m_1>m_2$,   was proved in \cite[Eqs. $(3.44)$ and $(3.45)$]{long2002four} that $r_{13}<r_{24}$. Hence $(q,m)$ is of  CCr type. From \cite[Corollary 2]{perez2007convex} if $m_1=m_2$ then the configuration is a square witch is a contradiction with the fact that $|q_1|\neq|q_2|$.

\emph{Case 2.} $|q_1|=|q_2|=|q_3|=|q_4|$. In this situation,  was proved in \cite{hampton2005co} that the configuration is the equal mass square.
\end{proof}

\section{Massless particle motion}\label{sec:mas-mot}
In this and next sections we suppose that the primary bodies are in a $T$-periodic rigid motion associated to a balanced collisionless CC $(q,m)$, i.e  $r(t)\equiv 1$ and according to remark following equation \eqref{eq:kepler.2.dim}, $\theta (t)=\sqrt{\lambda}t$. Without loss of generality,  we have assumed here that $\theta(0)=0$. For the particle, we suppose that it is moving on $L$, i.e. $x_0(t)=(0,0,z(t))$. From Theorem \ref{thm:prim}, $x_0$ is solution of \eqref{eq:newton}, if and only if $z(t)$ is solution of the autonomous equation

\begin{equation}\label{eq:eq_new_red}
 \ddot{z}=-\sum_{i=1}^n\frac{m_iz}{(s_i^2+z^2)^{3/2}},
\end{equation}
where $s_i=|q_i|$.

We will analyze all possible motions for the massless particle $x_0$. In particular we will see that all motion is periodic or is a scape trayectory. We will find that there exists $T_0$-periodic solutions for all $T_0$ in an interval  $(\sigma(q,m),+\infty)$. This fact implies that there exists an infinity quantity of periodic solutions for the entire $n+1$-body system.

The second order equation \eqref{eq:eq_new_red} is conservative, therefore solutions conserve the energy
\begin{equation}\label{eq:conser.energ}
E(z,v):=\frac{|v|^2}{2}-\sum_{i=1}^{n} \frac{m_i}{\left(s_i^2+z^2\right)^{\frac12}},
\end{equation}
i.e. $E(z(t),\dot{z}(t))$ is constant.

Following \cite{VladimirI.Arnold229} (see also \cite{marchesin2013spatial})  we introduce the next concepts.
\begin{definition}[Chazy, 1922]
 A solution $z(t)$ of \eqref{eq:eq_new_red} such that  $\lim\limits_{t\to\infty}z(t)=\infty$ is called hyperbolic for $t\to \infty$ when $\lim\limits_{t\to\infty}\dot{z}(t)= z_{\infty}\neq 0$ and is called parabolic if $\lim\limits_{t\to\infty}\dot{z}(t)=0$.
\end{definition}

The following theorem characterize all the possible motions for the massless particle.

\begin{theorem}\label{thm:prin_ine} We assume that $(q,m)$ is a balanced collisionless configuration and the primaries are in a rigid motion. Every solution of \eqref{eq:eq_new_red} is of some of the following types:
\begin{enumerate}
\item\label{1} Hyperbolic, when $E>0$,
\item\label{2} Parabolic, when $E=0$,
\item\label{3} Periodic, when $E_{min}:=-\sum_{i=1}^{n}\frac{m_i}{s_i}<E<0$.
\item\label{4} Equilibrium solution when $E=E_{min}$.
\end{enumerate}
\end{theorem}

\begin{proof}
We follow a standard argument for hamiltonian systems (see \cite{A}).

We consider the level sets $S(E)=\{(z,v):E(z,v)=E\}$, in the phase space $(z,v)$. An elementary analysis shows that
\begin{itemize}
 \item If $E\geq 0$ then $S(E)$ is the union of two bounded graphs. They are symmetric with respect to $z$-axis, each of which is contained  in some semiplane $v> 0$ or $v<0$. The $v$-positive branch is the graph of a function $v(E,z)$, which  is decreasing with respect to $|z|$. Moreover, $\lim\limits_{|z|\to \infty}v(E,z)=\sqrt{2E}$.

 \item For every $E\geq E_{min}$, the energy curve $S(E)$ cut the $v$-axis at the value $\pm(2E+2\sum_{i=1}^n m_is_i^{-1})^{\frac12}$.

 \item If $E_{min}<E<0$ then $S(E)$ is a simple closed curve symmetric with respecto to $z$ and $v$ axes.

 \item  An energy curve cut the $z$-axis, only in the case that $E<0$, at $\pm z_{E}$, where $z_E$ is the only positive solution of $-\sum_{i=1}^n m_i (s_i^2+z_{E}^2)^{-\frac12}=E$.
\end{itemize}

In the figure  \ref{fig:conf_esp} we show the phase portrait for a rhomboidal configuration with masses $m_1=m_3=1$ and $m_2=m_4=0.5$.

The function $\varphi(t)=(z(t),\dot{z}(t))$ solves the system $\dot{\varphi}(t)=F(\varphi(t))$, where \linebreak $F(z,v)=(v,-\sum_{i=1}^{n}m_iz (s_i^2+z^2)^{-3/2})$. The only fixed point of $F$ is $(z,v)=(0,0)$. Therefore, the level surfaces $S(E)$, with $E\neq E_{min}$, do not contain stationary points. Well know arguments imply that the trayectories $t\mapsto (z(t),\dot{z}(t))$ are defined in every time and   fill completely the connected component of the energy curves.

We observe that any solution $z$ crosses the $v$-axis. On the other hand if $E\geq 0$ and $v(E,0)>0$ ($v(E,0)<0$) then $z(t)$ is increasing (decreasing) with respect to $t$. If $z(t)$ remained bounded when $t\to +\infty$, then there would be the limit $\zeta_{\infty}:=\lim\limits_{t\to\infty}z(t)$. This would imply  that $(\zeta_{\infty},0)$ is a fixed point of $F$, which is a contradiction.  As a consequence, if $E\geq 0$ then $|z(t)|\to \infty$ when $t\to  +\infty$. Moreover $\lim\limits_{t\to +\infty}\dot{z}(t)=\pm\sqrt{2E}$.  From this we conclude that the trayectory is hyperbolic when $E>0$ and it is parabolic in the case $E=0$.

In the case that $E_{min}<E<0$ we have that the trayectory is contained in a closed curve, therefore it is a periodic orbit.

Finally if $E=E_{min}$ clearly we have that $z(t)\equiv 0$.
\begin{figure}[h]
\begin{center}
\includegraphics[scale=0.3]{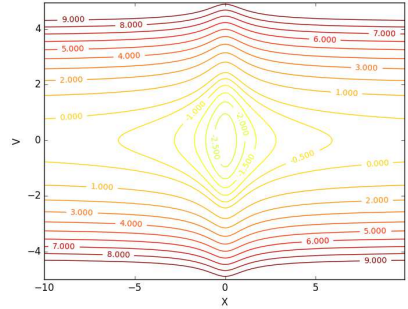}
\caption{Energy level for a rhomboidal configuration with masses $m_1=m_3=1$ and $m_2=m_4=0.5$.}\label{fig:energy}
\end{center}
\end{figure}
\end{proof}

\begin{theorem}\label{thm:prop.periodos}
We denote by $T_0(E)$ the minimal period for a solution of \eqref{eq:eq_new_red} with $E_{min}<E<0$. Then
\begin{enumerate}
 \item\label{it:T0.formula} for $z_E$ the only positive solution of $-\sum_{i=1}^n m_i (s_i^2+z_{E}^2)^{-\frac12}=E$
 \begin{equation}\label{eq:form.T0E-periodo}
 T_0(E)=2^{3/2}\int_{0}^{z_E} \left(E+\sum_{i=1}^n m_i (s_i^2+z^2)^{-\frac12}\right)^{-\frac12} dz,
 \end{equation}
 \item\label{it:T0.creciente} $T_0(E)$ is an increasing function.
 \item\label{it:T0.rango} $T_0\left((E_{min},0)\right)=(T_{min},+\infty)$, where  $T_{min}=2\pi\left(\sum_{i=1}^n\frac{m_i}{s_i^3} \right)^{-1/2}$.

 \end{enumerate}
\end{theorem}

\begin{proof}
Let $E_{min}<E<0$ and let $z(t)$ be the only solution with $z(0)=0$, $\dot{z}(0)>0$ and energy equal to $E$. Therefore $z(t)$ is $T_0(E)$-periodic. As a consequence of the symmetries of the equation we have that $z(T_0(E)/4)=z_E$. Then, taking account  \eqref{eq:conser.energ} we have that
\begin{equation*}
 \begin{split}
  \frac{T_0}{4}&=\frac{1}{\sqrt{2}}\int_0^{T_0/4}\left(E+\sum_{i=1}^n m_i (s_i^2+z^2)^{-\frac12}\right)^{-\frac12} \dot{z} dt\\
  &=\frac{1}{\sqrt{2}}\int_{0}^{z_E} \left(E+\sum_{i=1}^n m_i (s_i^2+z^2)^{-\frac12}\right)^{-\frac12} dz,
 \end{split}
\end{equation*}
and we have proved item \textit{1}. In order to prove item \textit{2} we note that
\begin{equation*}
 \begin{split}
  2^{-3/2}T_0(E)
  &=
    \int_0^{z_E}\left(\sum_{i=1}^n m_i \left((s_i^2+z^2)^{-\frac12}-(s_i^2+z_E^2)^{-\frac12}\right)\right)^{-\frac12}dz\\
  &=\int_0^{z_E} \left(z_E^2-z^2\right)^{-\frac12} f(z,z_E)dz\\
  &=\int_0^1 \left(1-u^2\right)^{-\frac12} f(z_Eu,z_E)du,
 \end{split}
 \end{equation*}
where \[f(z,z_E)=\left(\sum_{i=1}^n m_i \left\{(s_i^2+z^2)(s_i^2+z_E^2)\right\}^{-\frac12} \left\{(s_i^2+z^2)^{\frac12}+(s_i^2+z_E^2)^{\frac12} \right\}^{-1}\right)^{-\frac12}.\]
We note that $f(z_Eu,z_E)$ is a increasing function with respect to $z_E$ for $u\in [0,1]$ fix. This implies item \ref{it:T0.creciente}.

On the other hand
\begin{equation*}
 \lim\limits_{z_E \to 0}f(z_Eu,z_E)=\left(\sum_{i=1}^{n} \frac{m_i}{2s_i^3}\right)^{-\frac12} \quad \text{and}\quad  \lim\limits_{z_E \to +\infty}f(z_Eu,z_E)=+\infty.
\end{equation*}
Therefore, from the dominated convergence theorem and monotone convergence theorem we have that
\[\lim\limits_{E\to E_{min}}T_0=\lim\limits_{z_E\to 0}T_0=2\pi\left(\sum_{i=1}^{n} \frac{m_i}{s_i^3}\right)^{-\frac12}\quad \text{and}\quad \lim\limits_{E\to 0}T_0=\lim\limits_{z_E\to +\infty}T_0=+\infty.\]
Finally,  since $T_0=T_0(z_E)$ is continuous and increasing respect to $z_E$, we conclude the afirmation in the item \ref{it:T0.rango}.
\end{proof}

\begin{remark}
 It si posible to use the classical theory of hamiltonian systems (see \cite{A}) to derive the formula \eqref{eq:form.T0E-periodo} (see \cite{acinas2014estimates} for this approach in a related problem).
\end{remark}

\begin{remark}
Let us to show a second proof of item \ref{it:T0.rango} of Theorem \ref{thm:prop.periodos}.

The inequality $T_0>T_{min}$ is consequence of comparison Sturm's theorem applied to equations  $\ddot{z}+h(z)z=0$, where $h(z)=\sum_{i=1}^{n} m_i \left(s_i^2 +z^2\right)^{-3/2}$, and $\ddot{z}+\left(\sum_{i=1}^{n} m_i s_i^{-3}\right)z=0$. This prove that $T_0\left((E_{min},0)\right)\subset(T_{min},+\infty)$.

For the reverse inclusion  we  follow arguments of \cite{zhao2015nonplanar} and \cite{li2013characterization}, based on variational principles.

Let $T_0>T_{min}$. We consider the action integral
\[\mathcal{I}(z)=\int_0^{T_0}\frac12|\dot{z}|^2+\sum_{i=1}^n\frac{m_i}{\sqrt{s_i^2+z^2}}dt,\]
Then $T_0$-periodic solutions of \eqref{eq:eq_new_red} are critical points of $\mathcal{I}$ in the space $H^1(\mathbb{T},\rr)$, where $\mathbb{T}=\rr/T_0\mathbb{Z}$, of the functions  absolutely continuous, $T_0$-periodic with $\dot{z}\in L^2(\mathbb{T},\rr)$ (see \cite[Cor. 1.1]{Mawhin2010}). We prove the existence of critical points by means of the direct method of calculus of variations, i.e. we will prove that $\mathcal{I}$ has a minimum.  The functional $\mathcal{I}$ is not coercive in $H^1(\mathbb{T},\rr)$.  This deficiency is drawn with symmetry techniques (see \cite{David-2004}). The group $\mathbb{Z}_2$ acts on $H^1(\mathbb{T},\rr)$ according to the following assignments $(\bar{0}\cdot z)(t)=z(t)$ and $(\bar{1}\cdot z)(t)=-z(t+\frac{T_0}{2})$. The functional $\mathcal{I}$ is $\mathbb{Z}_2$-invariant, i.e. $\mathcal{I}(g\cdot z)=\mathcal{I}(z)$. We define the space of all $\mathbb{Z}_2$-symmetric (this simmetry is called the italian simmetry) funcions \[\Lambda(\mathbb{T},\
\mathbb{R}):
=\left\{ z\in H^1(\mathbb{T},\rr) | \forall g\in \mathbb{Z}_2 : z=g\cdot z \right\}.\]
The funciontal $\mathcal{I}$ restricted to $\Lambda$  is coercive. This fact follows from an obvious adaptation of Proposition 4.1 of \cite{David-2004}. We note that $F(z):=\sum_{i=1}^nm_i(s_i^2+z^2)^{-\frac{1}{2}}$ satisfies the condition $(A)$ in \cite[p. 12]{Mawhin2010}, then $\mathcal{I}$  is continuously differentiable and weakly lower semicontinuous on $H^1(\mathbb{T},\rr)$ (see \cite[p. 13]{Mawhin2010}). Therefore $\mathcal{I}$ has a minimum $z_0$ in $\Lambda(\mathbb{T},\mathbb{R})$. Then by the Palais' principle symmetric criticality,  $z_0$ is a critical point of $\mathcal{I}$ in $H^1(\mathbb{T},\rr)$ (see \cite{David-2004} and \cite{RichardPalais274}).

We use the second variation $\delta^2 \mathcal{I}$ in order to show  that $z_0\nequiv 0$. It is well known (see \cite[Th. 1.3.1]{jost1998calculus}) that if $z_0$ is a minimum of $\mathcal{I}$ on $H^1(\mathbb{T},\rr)$  then $\delta^2 \mathcal{I} (z_0,\varphi)\geq 0$ for all $\varphi\in H^1(\mathbb{T},\rr)$. In our case
\[\delta^2\mathcal{I}(0,\varphi)=\int_0^{T_0} |\dot{\varphi}|^2-\sum_{i=1}^{n}\frac{m_i}{s_i^3}\varphi^2 dt,\]
(see \cite[Eq. 1.3.6]{jost1998calculus}). In particular for $\varphi(t)=\sin (2\pi t/T_0)$ it follows from $T_0>T_{min}$  that
\begin{equation}\label{eq:form.delta2}
 \delta^2 \mathcal{I} (0,\varphi)=\left( \frac{4\pi^2}{T_0^2}-\sum_{i=1}^{n}\frac{m_i}{s_i^3} \right)\frac{T_0}{2}<0.
\end{equation}
It is sufficient  to guarantee that $z_0\equiv 0$ is not a minimum.

This second proof, unlike the first one, does not prove that $T_0$ is the minimum period for $z_0$. It could happen that $z_0$ had period $T_0/m$, with natural $m\in\mathbb{N}$. Because of Italian symmetry this $m$ should be odd.
\end{remark}

\begin{corollary}\label{cor:sol.periodica.sist.completo}
The complete $n+1$-body system has a infinity quantity of periodic solutions.
\end{corollary}
\begin{proof} We recall that $T$ denotes minimal period of primaries motion.
Let $l/m$  be a positive rational number with $Tl/m>T_{min}$. Then, there exists a solution of the entire system with period $lT$.
\end{proof}

\section{Synchronous solutions and pyramidal CC}\label{sec:sincro}

If the equation \eqref{eq:eq_new_red} has a $T$-periodic solution,  we say that the solution is \emph{synchronous}. In \cite{li2013characterization} was studied the problem of existence of synchronous solutions for $n$ equal mass primary bodies in a regular polygon configuration.

In this section we establish a relation between the existence of synchronous solutions and the concept of pyramidal central configuration (see \cite{fayccal1996classification,faycaltesis,ouyang2004pyramidal}).

\begin{definition}
A central configuration of $n+1$ mass point $q_0,\ldots,q_{n}$ in $\rr^{3}$  is called a pyramidal central configuration (PCC) if and only if $n$ points, we say $q_1,\ldots,q_n$, are in some plane $\Pi$ and $q_{0}\notin \Pi$.
\end{definition}

The following lemma was proved in \cite{ouyang2004pyramidal} (see also \cite{faycaltesis}).
\begin{lemma}[\cite{ouyang2004pyramidal}, Lemma 2.1]\label{lem:PCC}
 Let $q_0,\ldots,q_{n}$ be a PCC such that $m_{0}$ is off the plane containing $m_1,\ldots,m_n$. If $m_{0}>0$ then $m_{0}$ is equidistant from $m_1,\ldots,m_n$.
\end{lemma}

We remark that the condition $m_{0}>0$ is important in the previous Lemma. We will show below examples of two PCC with $m_{0}=0$ which do not satisfy the conclusion of  Lemma \ref{lem:PCC}.

\begin{proposition}\label{cor:sol.sincronica}
We assume that $q=q_1,\ldots,q_n$ is a balanced collisionless configuration and that the primaries are in a rigid motion. Then, there is a synchronous solution if and only if there exists $c\in \rr$ such that the points $(0,0,c),q_1,\ldots,q_{n}$ associated to the masses $0,m_1,\ldots,m_{n}$ form a PCC.
\end{proposition}

\begin{proof}
We start assuming that there exist a synchronous solution. As a consequence of the Theorem \ref{thm:prop.periodos}(\ref{it:T0.rango}) and the fact that $T^2=4\pi^2/\lambda$ we have that

 \begin{equation}\label{eq:lamdbda<suma.si3}
\lambda<\sum_{i=1}^n\frac{m_i}{s_i^3}.
 \end{equation}
 Since $\sum_{i=1}^{n}m_i\left(s_i^2+c^2\right)^{-3/2}\to 0$, when $c\to +\infty$, there exists $c\in \rr$ such that
 $ \sum_{i=1}^{n}m_i\left(s_i^2+c^2\right)^{-3/2}=\lambda$. Therefore
 \begin{equation}\label{eq:cond.CC1}
  -\sum_{i=1}^{n}\frac{m_i c}{\left(s_i^2+c^2\right)^{3/2}}=-\lambda c.
 \end{equation}
As $q_1,\ldots,q_n$ is a balanced configuration then
\begin{equation}\label{eq:cond.CC2}
   \sum_{i=1}^{n}\frac{m_i q_i}{\left(s_i^2+c^2\right)^{3/2}}= (0,0).
\end{equation}
The equations \eqref{eq:cond.CC1}, \eqref{eq:cond.CC2}  and the facts that $q_1,\ldots,q_n$ is a CC, with constant $\lambda$ complete the proof.
The proof of the reciprocal statement follows in a direct way.
\end{proof}

\begin{corollary}
We assume that $(q,m)$ is a balanced collisionless configuration and that the primaries are in a rigid motion. Then, there is a synchronous solution if and only if
 \begin{equation}\label{eq:ine_prin}
 \sum_{i<j}\frac{m_im_j}{r_{ij}}<\left(\sum_{i=1}^n\frac{m_i}{s_i^3}\right)\left(\sum_{i=1}^nm_is_i^2\right).
\end{equation}
\end{corollary}

\begin{proof}
The result is consequence of \eqref{eq:lamdbda<suma.si3} and the fact that $T^2=4\pi^2 \sum_{i=1}^{n}m_is_i^2/U$   (see \cite[p. 109]{JaumeLlibre276}).
\end{proof}

\begin{remark}\label{com:sincronicas}
We observe that if $(q,m)$  is a  balanced  CC  with constant $\lambda>0$, which satisfies \eqref{eq:ine_prin} and if $r,\mu>0$  then  $(rq,\mu m)$ is a CC with constant $\lambda \mu r^3$, and \eqref{eq:ine_prin} remains unchanged for the new positions $rq$ and the new masses $\mu m$. Consequently we can assume that any  length and  mass take any  desired value. Then the equation \eqref{eq:eq_new_red} has a synchronous solutions if and only if the same equation, with $(rq,\mu m)$ instead $(q,m)$, has a synchronous solution.
\end{remark}

The sufficiency of the condition $n\leq 472$ in the following corollary  was proved in \cite{li2013characterization}.

\begin{corollary}\label{cor:nleq472}
We suppose that $(q,m)$ is the equal masses regular polygon configuration  (this is an balanced CC). Then there exists a synchronous solution if and only if $2\leq n\leq 472$.
\end{corollary}

\begin{proof}
In this case $s_1=s_2=\cdots=s_n=:r$ and $m_1=m_2=\cdots=m_n=:M$. Then, from the law of cosines we obtain
\[
 \sum_{i<j}\frac{m_im_j}{r_{ij}}=\frac{nM^2}{4r}\sum_{j=1}^{n-1}\frac{1}{\sin\left(\frac{j\pi}{n}\right)}.
\]
Therefore the condition \eqref{eq:ine_prin} is equivalent to

\begin{equation}\label{eq:ine_prin_LShShao}
  \frac1n\sum_{j=1}^{n-1}\frac{1}{\sin\left(\frac{j\pi}{n}\right)}<4.
\end{equation}

This inequality was also derived by Li, J. et al. in \cite{li2013characterization}, where the authors proved (performing computer calculations) that inequality \eqref{eq:ine_prin_LShShao} holds true for $2\leq n\leq 472$. Let us prove that any other $n$ does not satisfies \eqref{eq:ine_prin_LShShao}.

Using that $1/\sin (x)$ is a convex function on $[0,\pi]$ and the composite trapezoid rule (see \cite{kincaid1991numerical}) we have that
\[
\begin{split}
 \int_{\frac{\pi}{n}}^{\frac{n-1}{n}\pi}\frac{1}{\sin (x)}dx&\leq \frac{\pi}{2n}\left\{ \frac{1}{\sin(\frac{\pi}{n})} + \frac{1}{\sin(\frac{n-1}{n}\pi)} +2\sum_{j=2}^{n-2}\frac{1}{\sin(j\frac{\pi}{n})} \right\}\\
 &=\frac{\pi}{n}\sum_{j=1}^{n-2}\frac{1}{\sin(j\frac{\pi}{n})}.
\end{split}
\]
Therefore
\[
\begin{split}
 \frac1n \sum_{j=1}^{n-1}\frac{1}{\sin\left(\frac{j\pi}{n}\right)}&\geq \frac{1}{\pi}\int_{\frac{\pi}{n}}^{\frac{\pi(n-1)}{n}}\frac{1}{\sin (x)}dx+\frac{1}{n\sin\left(\frac{n-1}{n}\pi\right)}\\
 &=\left.\frac{1}{2\pi}\log \left( \frac{1-\cos(x)}{1+\cos(x)}\right)\right|_{\frac{\pi}{n}}^{\frac{n-1}{n}\pi}+\frac{1}{n\sin\left(\frac{\pi}{n}\right)}\\
 &=\frac{1}{\pi}\left\{\log \left(\frac{1+\cos(\frac{\pi}{n})}{1-\cos(\frac{\pi}{n})}\right)+\frac{\pi/n}{\sin\left(\frac{\pi}{n}\right)}\right\}\\
 &=:f\left(\frac{\pi}{n} \right).
 \end{split}
\]
It is easy to see that $f(x)$ is a decreasing function on $(0,\pi/2)$. Moreover $f(\pi/842)\approx 4.0006>4$. Therefore, if $n\geq 842$ then $n$ does not satisfy inequality \eqref{eq:ine_prin_LShShao}. The validity of the inequality \eqref{eq:ine_prin_LShShao}, for $n\leq 841$, is easily checked using computer. This gives the result that the inequality holds only for $n \leq 472$.
\end{proof}

Our next objective is to verify that condition \eqref{eq:ine_prin} is satisfied for all balanced CC of 3-body or 4-body. Since we have already proved, in  Corollary \ref{cor:nleq472}, that \eqref{eq:ine_prin} holds for a equilateral triangle and square configurations of equal masses bodies, it only rest to prove, in virtue of Theorem \ref{thm:caracterizacion4}, the following result.

\begin{theorem}\label{thm:CC.3.4.satis.cond.adm}
The central configurations CCcl and CCr satisfy condition \eqref{eq:ine_prin}.
\end{theorem}

\begin{proof}

Let's start by analyzing the central configuration CCr. From the Remark \ref{com:sincronicas}, we can suppose without loss of generality that $ q_1 = -q_3 = (0, y) $ for $ 0<y<1 $, $ q_2 = -q_4 = (1,0) $. The condition \eqref{eq:ine_prin} becomes
\[\frac{m_1^2}{2y}+\frac{4m_1m_2}{\sqrt{1+y^2}}+\frac{m_2^2}{2}<\left(\frac{2m_1}{y^3}+2m_2\right) \left(2m_1y^2+2m_2\right).\]
As $m_1^2/(2y)<4m_1^2/y$, $m_2^2/2<4m_2^2$ and $4m_1m_2/\sqrt{1+y^2}<4m_1m_2/y^3$ (since $y<1$), we have that the inequality holds.

 Now we consider the central configuration CCl. From Remark \ref{com:sincronicas} again, we can suppose  that $q_1=-q_3=1$, $q_2=-q_4=x$ with $0<x<1$, and $m_1=m_3=\mu$, $m_2=m_4=1-\mu$, with $0<\mu<1$.  Then the inequality \eqref{eq:ine_prin} becomes
\[\frac{2\mu(1-\mu)}{1-x} +\frac{2\mu(1-\mu)}{1+x}+\frac{\mu^2}{2}+\frac{(1-\mu)^2}{2x}<4\mu^2+4\mu(1-\mu)x^2+\frac{4\mu(1-\mu)}{x^3}+\frac{4(1-\mu)^2}{x}.\]
As $ \mu ^ 2/2 <4 \mu^2$ and $ (1-\mu)^2/(2x)< 4(1-\mu)^2/x $ it is sufficient to show that
\[\frac{2\mu(1-\mu)}{1-x} +\frac{2\mu(1-\mu)}{1+x}<\frac{4\mu(1-\mu)}{x^3},\]
and this is equivalent to see that
\begin{equation}\label{eq:ineq.4.cuerpos.alinea}
\frac{x^3}{1-x^2}<1.
\end{equation}
The values of $x$ involved in the above inequality are such that the configuration of positions $(-1,-x,x,1)$ and masses $(\mu,1-\mu,1-\mu,\mu)$ is central. It was shown in \cite{moulton1910straight} that given a mass $\mu$ there is only one value of $x$ satisfiying this condition (see also \cite{shoaib2011collinear}). So, we can define $x(\mu)$ as such value of $x$.  We note that $h(x)=x^3/(1-x^2)$ is an increasing function with respect to $x\in (0,1)$ and $h(x)< 1$ for $x\in (0,3/4)$. Hence, if we could prove that $x(\mu)$ is a decreasing function and
\begin{equation}\label{eq:ineq.4.cuerpos.lim0}
\lim\limits_{\mu\to 0}x(\mu)<3/4
\end{equation}
we would have justified \eqref {eq:ineq.4.cuerpos.alinea}.

Let's first prove that $x(\mu)$ is a decreasing function. Eliminating $\lambda$ from the equations \eqref{eq:def.CC} and replacing $q_j$ and $m_j$ by their expressions in $x$ and $\mu$ we get
\[\frac{\mu}{4} - \frac{\mu}{x \left(x + 1\right)^{2}} + \frac{\mu}{x \left(- x + 1\right)^{2}} + \frac{- \mu + 1}{\left(x + 1\right)^{2}} + \frac{- \mu + 1}{\left(- x + 1\right)^{2}} - \frac{1}{x^{3}} \left(- \frac{\mu}{4} + \frac{1}{4}\right) = 0\]
which is equivalent to
 \[\mu=- \frac{8 x^{5} - x^{4} + 8 x^{3} + 2 x^{2} - 1}{\left(x - 1\right) \left(x + 1\right) \left(x^{5} - 9 x^{3} + x^{2} - 1\right)}.\]
Therefore
\[
 \frac{d\mu}{dx}=\frac{x^{2} \left(16 x^{9} - 3 x^{8} + 32 x^{7} + 12 x^{6} - 304 x^{5} - 2 x^{4} + 44 x^{2} - 51\right)}{\left(x - 1\right)^{2} \left(x + 1\right)^{2} \left(x^{5} - 9 x^{3} + x^{2} - 1\right)^{2}}.
\]
Since $44x^2<51$ and $16 x^{9}  + 32 x^{7} + 12 x^{6} < 304 x^{5}$ for $x\in (0,1)$ then $d\mu/dx<0$ in the interval $ (0,1)$. Which, in turn, implies that $x$ is decreasing respect to $\mu$.

Let's see now that \eqref{eq:ineq.4.cuerpos.lim0} holds. When $\mu$ goes to $0$, $x(\mu)$ converges to the only solution  in the interval $(0,1)$ of  equation $8 x(0)^{5} - x(0)^{4} + 8 x(0)^{3} + 2 x(0)^{2} - 1=0$.  Then  $ 8 x(0)^{3} -1< 0$ which implies that $x(0)<3/4$ as we wanted to prove.
\end{proof}

\begin{remark}
As consequence of previous results there exist five-body $PCC's$ with $m_1,\ldots,m_4$  in a CCcl or CCr configuration and the mass $m_0=0$ is in the  line perpendicular to the  plane containing $m_1,\ldots,m_4$ and passing by the center of mass. These are examples of $PCC's$ wich does not verify the conclusion of Lemma \ref{lem:PCC}.
\end{remark}

\begin{corollary}
For all  balanced CC of 3-body or 4-body, the  problem $P$ has a synchronous solution.
\end{corollary}

\section{Non balanced central configurations}

The following result shows a necessary condition for that  a non-balanced CC allows a solution of the problem $P$.

\begin{theorem}\label{thm:no.admisible.movimiento}
We suppose that $(q,m)$ is a non balanced CC and  that the primaries are in a homographic motion, i.e.  equation \eqref{eq:x_j=rtQtq_j} is satisfied. Assume that the massless particle is moving on the $z$-axis with position vector $x_0(t)=(0,0,z(t))$. Then, some of the following statements are satisfied:
\begin{enumerate}
 \item\label{it:z==0} The massless particle is in a stationary motion and
 \begin{equation}\label{eq:acel.centrmasa=0}
  \sum_{i=1}^{n}\frac{m_iq_i}{s_i^3}=0,
 \end{equation}
 i.e. the positions $0,q_1,\ldots,q_n$ and the masses $0,m_1,\ldots,m_n$ are in a CC.
 \item\label{it:z=r} The $n+1$-body system is in a homothetic motion. i.e. $Q(\theta(t))$ in the equation \eqref{eq:x_j=rtQtq_j} is the identity matrix and $z(t)=cr(t)$, for some constant $c$. Moreover, the configuration $q_0,\ldots,q_n$ is a PCC, where $q_0=(0,0,c)$ and $m_0=0$.
\end{enumerate}
\end{theorem}

\begin{proof}
 We recall the definition of the function $f$ and line $L$ from the Section \ref{sec:admisible.configuraciones}.

The fact that the massless particle is moving on $L$, is equivalent to the condition $f(t,x_0(t))\in L$ for all $t$, which, instead, is equivalent to the equality
\begin{equation}\label{eq:z/r}
 \sum_{i=1}^n\frac{m_ir(t)Q(\theta (t))q_i}{\left(r(t)^2|q_i|^2+z(t)^2\right)^{3/2}}=0,
\end{equation}
for every $t\in \rr$.

With the same notation and following similar reasoning that in the proof of Theorem \ref{thm:prim}, we prove that
\begin{equation}\label{eq:sumr.sumFr_miqi=0}
\sum_{j=1}^k\left\{\frac{1}{(s_j^{2}+(z(t)/r(t))^2)^{3/2}}\sum_{i\in F_j}m_iq_i\right\}=0.
\end{equation}
If $z(t)/r(t)$  would be a non constant function then previous equation and Lemma \ref{lem:1} would imply that $q$ is balanced, which is a contradiction. Hence there exist $c\in \rr$ such that $z(t)=cr(t)$. Now, we have two cases.

\emph{Case 1:} $c=0$. Then $z\equiv 0$ and \eqref{eq:acel.centrmasa=0} follows from \eqref{eq:z/r}.

\emph{Case 2:} $c\neq 0$. From equation \eqref{eq:eq_new_red}, the Kepler equations \eqref{eq:kepler.2.dim}, and the fact that $z(t)=cr(t)$ we have that
\begin{equation}\label{eq:kepler=CC}
 -\frac{1}{r(t)^2}\sum_{i=1}^{n}\frac{m_i}{(s_i^2+c^2)^{3/2}}=-\frac{\lambda}{r(t)^2}+r(t)\dot{\theta}(t)^2.
\end{equation}
The second equality in \eqref{eq:kepler.2.dim} implies the Kepler's second law, i.e. there exists $d\in\rr$ such that $r^2\dot{\theta}\equiv d$. Replacing $\dot{\theta}$ in equation \eqref{eq:kepler=CC} and multiplying by $r(t)^3$ we obtain
\begin{equation}\label{eq:r.sum-lambd=d}
-r(t)\left(\sum_{i=1}^{n}\frac{m_i}{(s_i^2+c^2)^{3/2}}-\lambda\right)=d^2.
\end{equation}
Therefore, if $d\neq 0$ then $\dot{r}(t)\equiv 0$, and this implies $\dot{z}(t)\equiv 0$. As $z(t)$ is a constant function and it solves equation \eqref{eq:eq_new_red}, then $z(t)\equiv 0$. Hence we are in the case 1 again. Consequently we suppose $d=0$. Therefore $\theta(t)\equiv cte$ and the motion is homothetic. From \eqref{eq:sumr.sumFr_miqi=0} and \eqref{eq:r.sum-lambd=d} we deduce that in this new situation equation \eqref{eq:cond.CC1} and \eqref{eq:cond.CC2} hold. This, as in the proof of Proposition \ref{cor:sol.sincronica}, implies the desired result.
\end{proof}

\begin{example}
 We present an example of a $3+1$-body system satisfiying the situation described in the item \ref{it:z==0} of the Theorem \ref{thm:no.admisible.movimiento}, i.e. $(q,m)$ is a non balanced CC and $z(t)\equiv 0$. For this, it is sufficient to find a 4-body CC with a zero mass body located in the center of mass.

 We start with a Euler's collinear central configuration formed by three primary bodies of masses $m_1 = 4-\mu$, $m_2 = 2 + \mu$ and $m_3 = 1$, where $0<\mu<1$, and positions, respect to a convenient $1$-dimensional coordinate system,  given by $q_1 = 0$, $q_2 = 1$ and $q_3 = 1 + r$. It is know (see \cite{Moeckel:2014}) that $r$ is the only positive solution of
\[
p(r,\mu):=6 r^{5} +\left(16- \mu \right) r^{4}  +  \left( 14- 2 \mu \right) r^{3}- \left( \mu + 5\right)  r^{2}-\left( 2 \mu + 7\right) r - \mu - 3=0.
\]
Since  $p(0,\mu)=-\mu-3$ and $p(1,\mu)=-7\mu+21$ then $r=r(\mu)\in (0,1)$, for all $0<\mu<1$.
In this case the center of mass $C=C(\mu)$ is equal to $(\mu+r+3)/7$, so $C\in (0,1)$.

We consider a massless particle with coordinate $x$. The acceleration resulting from the action of the gravitational field is equal to
\[
f(x)= - \frac{4-\mu }{x^{2}}+\frac{\mu + 2}{\left(- x + 1\right)^{2}} + \frac{1}{\left(r - x + 1\right)^{2}}.
\]
Note that the right hand side of the previous equation is an increasing function that tends to $-\infty$ when $x$ goes to 0, and tends to $+\infty$ when $x$ goes to 1, so there is a unique point $\bar{x}=\bar{x}(\mu)\in (0,1)$ such that the equality $f(\bar{x})=0$ holds. This point is an equilibrium for the gravitational field generated for the primaries.

Let's see that there exists $ \mu \in (0,1) $ such that $ C(\mu) = \bar{x} $, i.e. $f(C)=0$. For this purpose, since $C$ is a continuous function with respect to $\mu$,  we show that $f$ changes of sign on  $(0,1)$.  The function $f(x)$ can be written as $$f(x)=\frac{Nf(x)}{Df(x)},$$ where $Df(x)=x^{2} \left(x - 1\right)^{2} \left(r - x + 1\right)^{2}$. Note that  $Df(x)>0$ for all $x\in (0,1)$. If we consider $\mu=0$ and compute $Nf(C)$ we have that
\[Nf(C)=\frac{r^{4}}{2401} + \frac{1514 r^{3}}{2401} + \frac{2245 r^{2}}{2401} + \frac{1110 r}{2401} + \frac{333}{2401}>0,\]
on the other, hand if  $\mu=1$ then
\[Nf(C)=- \frac{71 r^{4}}{2401} + \frac{1486 r^{3}}{2401} + \frac{401 r^{2}}{2401} - \frac{1480 r}{2401} - \frac{592}{2401}<0,\]
because $0<r<1$.
\end{example}

\begin{remark}
 The following question is posed. Is there some non balanced central configuration $(q,m)$ such that the $n+1$-body system perform the motion described in Theorem \ref{thm:no.admisible.movimiento}(2)?
\end{remark}

\def\cprime{$'$}

\end{document}